\documentclass[journal]{IEEEtran}
\IEEEoverridecommandlockouts

\makeatletter
\def\markboth#1#2{\def\leftmark{\@IEEEcompsoconly{\sffamily}\MakeUppercase{\protect#1}}%
\def\rightmark{\@IEEEcompsoconly{\sffamily}\MakeUppercase{\protect#2}}}
\makeatother

\usepackage[pdftex]{graphicx}
\usepackage{pgfplots}
\def\func1(#1,#2){exp(-(#1-#2)^2/(0.01))}
\usepackage{tikz}
\usepackage{psfrag}
\usetikzlibrary{external}
\tikzset{external/system call={pdflatex \tikzexternalcheckshellescape 
                                        -halt-on-error
                                        -interaction=batchmode 
                                        -jobname "\image" "\texsource"
                                        && inkscape \image.pdf --export-eps=\image.eps --export-ps-level=3}}
\usepackage{verbatim}
\usepackage{amsfonts}
\usepackage{amssymb}
\usepackage[font=small]{caption}
\usepackage{amssymb}
\usepackage{amsmath}
\usepackage{amsthm}
\usepackage{cancel}
\usepackage{epsfig}
\usepackage{varioref}
\usepackage{cite}
\usepackage{blkarray}
  \usepackage[pdftex]{graphicx}
  \usepackage{subfigure}
  \usepackage{epstopdf}
   \graphicspath{{./figures/}}
   \DeclareGraphicsExtensions{.eps,.ps,.png}
   
\usepackage{xcolor}
\usepackage{enumerate}

\usepackage{stfloats}
\usepackage{balance}

\newcommand{\Sm}{\mathbf{S}}
\newcommand{\U}{\mathbf{U}}

\newcommand{\V}{\mathbf{V}}
\newcommand{\A}{\mathbf{A}}
\newcommand{\B}{\mathbf{B}}

\newcommand{\I}{\mathbf{I}}
\newcommand{\Y}{\mathbf{Y}}

\newcommand{\Z}{\mathbf{Z}}

\newcommand{\one}{\mathbf{1}}

\newcommand{\h}{\mathbf{h}}
\newcommand{\x}{\mathbf{x}}

\newcommand{\s}{\mathbf{s}}

\newcommand{\vv}{\mathbf{v}}
\newcommand{\uu}{\mathbf{u}}
\newcommand{\y}{\mathbf{y}}

\newcommand{\z}{\mathbf{z}}
\newcommand{\ab}{\mathbf{a}}

\newcommand{\tr}{\textnormal{tr}}
\newcommand{\vstack}{\textnormal{vec}}

\newcommand{\Ex}[2]{{\textnormal{E}_{#1}\left[#2\right]}}

\newcommand{\CInf}[3]{{\textnormal{I}\left(#1;#2|#3\right)}}
\newcommand{\Inf}[2]{{\textnormal{I}\left(#1;#2\right)}}

\newtheorem{definition}{Definition}
\newtheorem{lemma}{Lemma}
\newtheorem{corollary}{Corollary}
\newtheorem{example}{Example}
\theoremstyle{plain}
\newtheorem{remark}{Remark}

   \definecolor{blueH3}{rgb}{0,.5,1}
   \definecolor{blueH2}{rgb}{0,0.25,0.75}
   \definecolor{blueH1}{rgb}{0,0,0.5}   
   \definecolor{grayOldText}{rgb}{.5,.5,.5}
   \definecolor{VCobalt}{HTML}{005682}
   \definecolor{TZTeal}{HTML}{008080}
   \definecolor{TZTealfaded}{HTML}{F0FFFF}
   \definecolor{KYJade}{HTML}{008151}
   \definecolor{ARust}{HTML}{a10000}
   \definecolor{FFucsia}{HTML}{7000c3}   
   \definecolor{Tangerine}{HTML}{d45500}
   
\newcommand{\CASE}[1]{\STATE \textbf{case} #1\textbf{:} \begin{ALC@g}}
\newcommand{\ENDCASE}{\end{ALC@g}}

\newcommand{\DEFAULT}{\STATE \textbf{default:} \begin{ALC@g}}
\newcommand{\ENDDEFAULT}{\end{ALC@g}}
\newcommand{\DEFAULTLINE}[1]{\STATE \textbf{default:} }

\newcounter{MYtempeqncnt}

\makeatletter
\newcommand\remembertext[2]{
  \immediate\write\@auxout{\unexpanded{\global\long\@namedef{mytext@#1}{#2}}}%
  #2%
}
\newcommand\recalltext[1]{%
  \ifcsname mytext@#1\endcsname
    \@nameuse{mytext@#1}%
  \else
    ``??''
  \fi
}

\makeatother
\usepackage{colortbl}
\usepackage[most]{tcolorbox}
\newtcolorbox{tcbremark}{
  enhanced jigsaw,
  oversize,
  rightrule=0pt,
  toprule=0pt,
  bottomrule=0pt,
  colback=white,
  arc=0pt,
  outer arc=0pt,
  title style={white},
  fonttitle=\color{black}\bfseries,
  titlerule=0pt,
  bottomtitle=0pt,
  top=0pt,
  bottom=0pt,
  left=5pt,
}

\newcounter{rcnt}
\newcounter{ccnt}

\newlength{\ansspace}
\newlength{\stdleftskip}
\newlength{\stdrightskip}
\newlength{\citeskip}
\newtheorem{theorem}{Theorem}

\begin{document}
\sloppy

\title{The SIMO Block Rayleigh Fading Channel Capacity Scaling with Number of Antennas, Bandwidth and Coherence Length}

%

\author{Felipe Gomez-Cuba$^1$ \thanks{
    $^1$AtlanTTic, Universidade de Vigo, Spain. Email: \texttt{gomezcuba@gts.uvigo.es} } }

\maketitle

\begin{abstract}
This paper studies the capacity scaling of non-coherent Single-Input Multiple-Output (SIMO) independent and identically distributed (i.i.d.) Rayleigh block fading channels versus bandwidth ($B$), number of receive antennas ($N$) and coherence block length ($L$). In non-coherent channels (without Channel State Information --CSI) capacity scales as $\Theta\left(\min(B,\sqrt{NL},N)\right)$. This is achievable using Pilot-Assisted signaling. Energy Modulation signaling rate scales as $\Theta\left(\min(B,\sqrt{N})\right)$. If $L$ is fixed while $B$ and $N$ grow, the two expressions grow equally and Energy Modulation achieves the capacity scaling. However, Energy Modulation rate does not scale as the capacity with the variable $L$. The coherent channel capacity with a priori CSI, in turn, scales as $\Theta\left(\min(B,N)\right)$. The coherent channel capacity scaling can be fully achieved in non-coherent channels when $L\geq\Theta(N)$. In summary, the channel coherence block length plays a pivotal role in modulation selection and the capacity gap between coherent and non-coherent channels. Pilot-Assisted signaling outperforms Energy Modulation's rate scaling versus coherence block length. Only in high mobility scenarios where $L$ is much smaller than the number of antennas ($L\ll\Theta(\sqrt{N})$), Energy Modulation is effective in non-coherent channels.
\end{abstract}

\section{Introduction}
\label{sec:introduction}

Internet traffic demand is expected to continue to grow exponentially \cite{Cisco2019}. Technologies such as the recent 5G cellular \cite{3GPPNRoverall16} and WiFi6 Wireless Local Area Network \cite{khorov2019tutorial,8319416} standards are upgraded regularly. The increase of the transmission bandwidth and number of antennas in each new standard generation are major factors in the increase of capacity of wireless systems \cite{Dohler2011,marzetta2006much}. In addition, the new standards support ``flexible waveforms'' that can dynamically modify the transmit bandwidth and number of antenna ports used by different users. As more spectrum and antennas are added to wireless standards, efficient channel estimation and the impact of channel state information (CSI) uncertainty has received a lot of attention. On the one hand, CSI limitations are one key factor in the performance of coherent receivers in \textit{massive MIMO} \cite{marzetta2006much}. On the other hand, recent results for the wideband channel have sparked new interest in non-coherent signaling techniques \cite{journals/tit/MedardG02}. Specifically, there has been significant recent literature proposing practical modulation schemes based on non-coherent Energy Detection (see \cite{9445644} and references therein). Overall, as wireless standards evolve, it is very important to understand the relation between the waveform dimensions (specifically bandwidth and number of antenna ports), the time-varying channel uncertainty, and the engineering choice between Pilot Assisted (PA) or non-coherent signaling schemes.

From an Information Theoretic point of view, the ultimate impact of channel uncertainty can be revealed by comparing the capacity of the \textit{coherent channel}, with the assumption that CSI is a priori revealed to the receiver; and the \textit{non-coherent channel}, in which CSI is not a priori available. On the transmitter side, transmitter CSI or feedback may also be assumed or not. We leave the topic of transmitter CSI out of the scope of this paper, in which we focus on comparing the capacity with and without receiver CSI, without feedback in both cases. In the non-coherent channel, it is possible to either use pilot sequences to estimate the channel, or to employ non-coherent signaling schemes. Thus, note that the non-coherent channel capacity also upper bounds the maximum rate of coherent receivers that rely on channel estimation (the name ``non-coherent channel'' is a bit counter-intuitive in this sense). While the coherent channel capacity is well-known \cite{goldsmith2005book}, the non-coherent channel capacity is not yet fully characterized. Important results have obtained the Degrees of Freedom (DoF) and Diversity-Multiplexing Trade-off (DMT) in the high-SNR regime \cite{866662,zhengoptimal,zheng2002communication,zheng2002diversity}, the wideband capacity in the limit as bandwidth goes to infinity \cite{journals/tit/TelatarT00,journals/tit/MedardG02,journals/tit/Verdu02,Medard2005,Zheng2007noncoherent,Ray2007noncoherent,Sethuraman2009,journals/twc/LozanoP12,fgomezUnified}, some properties of the probability density of the optimal channel input \cite{Marzetta1999,Medard2005,Perera2006,Chowdhury2014a,7541626,1337103,8826443}, and optimal pilot strategies to maximize the ``achievable rate'' of specific PA receivers \cite{Hassibi2003,Lozano2008,Jindal2010}.

In this paper we compare the capacity \textit{scaling} for coherent and non-coherent wideband Single-Input Multiple-Output (SIMO) block Rayleigh fading channels. Within the scope of achievable rates in a non-coherent channel, we also compare two signaling strategies: channel estimation with pilots vs energy detection. We focus on the single transmit antenna case for better comparison with closely related results in \cite{8826443}, detailed in the next paragraphs. The assumption of a wideband block fading model represents a reasonable abstraction for some modern wireless systems that use multi-carrier modulations and time-frequency-slotted frame allocations \cite{3GPPNRoverall16,khorov2019tutorial,8319416}. With the aim to inform the selection of waveform parameters, spectrum policies and antenna array dimensions, our analysis studies the effect on capacity growth of the following three major variables:
\begin{enumerate}
 \item \textbf{The Channel Coherence Time:} refers generally to the interval of time during which the random channel displays correlation. This concept is subject to different definitions under different channel models. In i.i.d. block-fading models, we define the \textbf{channel coherence length} as the length of time, in channel uses, that the channel remains constant. The block-fading channel coherence length has been shown to affect the non-coherent channel capacity-achieving input distribution \cite{Marzetta1999,Perera2006}, DMT \cite{zheng2002diversity}, and optimal pilot schemes \cite{Hassibi2003}. The intuition is that, if the channel remains constant for a longer time, the relative size of the pilot overhead compared to the total data codeword length decreases. Thus, the non-coherent channel capacity is expected to converge to the coherent channel capacity as the channel coherence length increases. This holds for all existing results such as the optimal input \cite{Marzetta1999} and DMT \cite{zheng2002diversity}.

 \item \textbf{The Bandwidth:} allows to increase the discrete time symbol rate, but it also decreases the Signal to Noise Ratio (SNR) when the thermal noise has constant power spectral density. The wideband non-coherent channel capacity is closely related to the low-SNR asymptotic analysis of capacity \cite{journals/tit/Verdu02,Medard2005,Zheng2007noncoherent,Ray2007noncoherent,Sethuraman2009}. M\'edard and Gallager \cite{journals/tit/MedardG02} showed that ``overspreading'' occurs when the channel input has finite fourth moment and bandwidth is too large, causing the achievable rate in the non-coherent channel to begin decreasing, rather than increasing monotonically as in the coherent case, if bandwidth is increased excessively. For Rayleigh fading, the ``critical bandwidth'' can be explicitly calculated \cite{journals/twc/LozanoP12}. The use of ``peaky'' input signals with an infinite fourth moment has been proposed to overcome this limitation \cite{Medard2005,Zheng2007noncoherent,Ray2007noncoherent}. Nevertheless, it can be shown that peaky signaling does not actually remove the critical bandwidth limit, but rather moves it to the time domain under a constant ``critical bandwidth occupancy'' limitation \cite{1337103,fgomezUnified}. The bandwidth overspreading threshold has been shown to grow with the square root of the coherence block length \cite{Ray2007noncoherent,journals/twc/LozanoP12,fgomezUnified}, showing that both parameters are closely related and that also in this sense the non-coherent channel capacity converges to the coherent channel capacity when the coherence length increases. The recent popularity of mmWave frequency bands increases interest in communications with very large bandwidth and antenna arrays \cite{Du2017,Ferrante2016}. Although in these bands the assumption of rich scattering with a Rayleigh fading distribution may not be applicable, concern about these phenomena in sparse multipath mmWave channels is justified, as overspreading occurs in other channel models, including sparse multipath channels \cite{journals/tit/TelatarT00}. We remark that the results in our paper apply to the i.i.d. Rayleigh channel, and while mmWave channels are subject to the problem of overspreading, engineers must be cautious about the channel model difference before applying our conclusions in a mmWave context. A discussion of literature on multiple mmWave and ultra-wideband channel modeling frameworks and their relation to our result is provided in Appendix \ref{ap:chanmod}.
 
 \item \textbf{The Number of Antennas:} has been shown to increase the DMT of the channel
 \cite{866662,zhengoptimal,zheng2002communication,zheng2002diversity}, enabling dramatic increases in the array gain or spatial multiplexing. In recent years, \textit{Massive MIMO} proposed the use of a large number of antennas \cite{marzetta2006much,Bjornson2016,hoydis2018,marzetta2013special}. Massive MIMO enables a very large array combining gain or diversity gain, mitigating pathloss and fading. Moreover the asymptotic properties of large matrices permit to implement efficient transceivers \cite{journals/ftcit/TulinoV04}. The DMT and optimal input distribution results in \cite{Marzetta1999,zheng2002communication} imply that, assuming there is no feedback, increasing the number of transmit antennas above the channel coherence length does not increase the non-coherent channel capacity. For this reason, analyses of the capacity scaling assuming a fixed coherence length have to focus on studying a large number of \textbf{receive} antennas. Chowdhury \textit{et al} \cite{Chowdhury2014a} showed that, as the number of receive antennas goes to infinity, an Energy Modulation (EM) constellation in a frequency-flat non-coherent SIMO channel can achieve the same capacity scaling as a coherent channel. This result was extended to the frequency selective case with a large bandwidth in \cite{8826443}. Unlike in this paper, the scaling result in \cite{8826443} considered scaling for both the bandwidth and the number of antennas, but assumed a fixed coherence length. This means that the result in \cite{8826443} side-stepped one key parameter in the study of the capacity gap between coherent and non-coherent channels. Note that in our paper we assume that the coherence block length scales as well, and thus we could potentially support a scaling number of transmit antennas too. However, we adopt the SIMO channel model to provide a much more clear comparison with \cite{8826443}, as well as for the sake of clarity of the model and problem tractability. We leave the extension of our result to MIMO for future work.
\end{enumerate}

Most Massive MIMO analyses assume a fixed bandwidth \cite{zheng2002communication,zheng2002diversity,marzetta2006much,Chowdhury2014a} and most wideband analyses assume a fixed number of antennas \cite{journals/tit/Verdu02,Ray2007noncoherent,Sethuraman2009,journals/tit/TelatarT00,journals/twc/LozanoP12,fgomezUnified}. However, in practice these two parameters, as well as the coherence length, display a non-trivial interplay between them. The DMT is related to high-SNR capacity analyses, whereas wideband capacity calculations focus on low-SNR regimes, making the combination of results from both families of literature challenging. The analysis in \cite{8826443} studied capacity scaling when \textit{both} the bandwidth and the number of antennas grow jointly with a fixed coherence length \cite{8826443}. In brief words, the main principle in \cite{8826443} and our model is a two-regime rate analysis that combines the DoF-limited narrowband and the power-limited wideband rate scaling: if the bandwidth is below a certain threshold, capacity scaling is in a high-SNR-like regime in which rate scales with the symbol rate and the DoF of the channel. Conversely, when the bandwidth exceeds the overspreading threshold, capacity scaling adopts a low-SNR-like regime in which rate scales with the total energy arriving at the receiver (possibly featuring the receive array gain). Moreover, this bandwidth threshold itself changes with the number of antennas and coherence block length. The result in \cite{8826443} showed that, when the coherence length is fixed, the non-coherent channel capacity cannot scale faster than the minimum between the bandwidth and the square root of the number of receive antennas, failing to match the linear scaling with the number of receive antennas of the coherent channel capacity in the wideband regime \cite{8826443}. Moreover, \cite{8826443} shows that EM achieves the joint capacity scaling in the non-coherent channel. As the analysis in \cite{8826443} assumes a constant coherence length, the main goal of our paper is to complete the model by fully taking into account the channel coherence length. This will show that changes in the channel coherence length play a very important role in capacity scaling that was not captured by the result of \cite{8826443}.

In this paper we characterize the capacity of coherent and non-coherent channels as a \textbf{three dimensional} function of bandwidth, number of receive antennas, and channel coherence length. To characterize the capacity scaling, we consider that these three parameters grow to infinity jointly. We revise the achievable rate of the EM scheme proposed in \cite{8826443}, and introduce an alternative PA signaling scheme that outperforms EM in terms of achievable rate scaling. 
In consequence, our result shows that EM is no longer generally optimal, and depending on the channel coherence length, PA schemes can provide a greater rate.
In addition, we compute an upper bound to non-coherent capacity scaling, which shows that our PA scheme always scales optimally, i.e. it can always achieve the capacity scaling, whereas EM cannot. 

Our main result shows that, when the bandwidth is smaller than a certain threshold, the coherent channel and the non-coherent channel can achieve the same capacity scaling. In general, this is the ``high-SNR-like'' regime in which capacity grows linearly with the ``pre-log'' term, i.e. in our case the bandwidth. The critical bandwidth threshold scales with the minimum between $i)$ the square root of the product of the number of antennas and the coherence length, or $ii)$ just the number of receive antennas. When the bandwidth is greater than the critical bandwidth threshold, the non-coherent channel enters the ``low-SNR-like'' regime in which the capacity grows linearly with the received energy, i.e. in our case with the minimum between the square root of the product of the number of antennas and the coherence length, or the number of receive antennas. Therefore, for large bandwidths, when the channel coherence length is large in comparison to the number of antennas, in a scaling sense, the conditions of the converse proof in \cite{8826443} are no longer applicable and EM signaling cannot achieve the non-coherent channel capacity scaling. Moreover, when the channel coherence length scales faster than the number of antennas, the non-coherent wideband massive SIMO channel achieves the same capacity scaling as the coherent channel, as should be expected, and PA channel estimation achieves the capacity scaling in all regimes. 

If we assume a fixed coherence length in our result, the critical bandwidth threshold scales with the square root number of antennas and the non-coherent channel capacity scaling reproduces the result of \cite{8826443}. Moreover, the EM scheme of \cite{8826443} can achieve the non-coherent capacity in this particular case. However, our result generalizes \cite{8826443} showing that, as the coherence length increases, the non-coherent capacity scaling approaches that of the coherent channel, PA schemes can achieve this rate scaling, and EM schemes remain limited to the same rate scaling as in \cite{8826443}, failing to keep up with the capacity.

The rest of this paper is structured as follows. Section \ref{sec:model} describes the system model. Section \ref{sec:knownres} describes a selection of prior known results relevant to our analysis. Section \ref{sec:mainres} contains the main results of our analysis. Section \ref{sec:ub} presents the capacity scaling upper bound that proves the converse of the main result. Section \ref{sec:pa} presents the PA scheme that achieves the capacity scaling result. Section \ref{sec:fem} presents our revision of the EM encoding scheme. Section \ref{sec:sim} demonstrates the main results in simulations. And finally, Section \ref{sec:conclusions} provides the conclusions.

\subsection{Notation}

Lowercase bold letters $\ab$ denote vectors and uppercase bold letters denote matrices $\A$. We use $\A^*$, $\A^T$ and $\A^H$ to denote the conjugate, transpose and Hermitian matrices of $\A$, respectively. The Kro\"enecker product between two matrices $\A$ and $\B$ is
denoted as $\A\otimes\B$.
We use $a_i$ to denote the $i^{th}$ element of vector $\ab$ and $A_{ij}$ for the $i,j$-th element in a matrix $\A$. $\vstack(\A)$ denotes the ``vector stacking'' of matrix $\A$, defined as a vector containing all the coefficients of $\A$ in column-first order, i.e. the $(i,j)$-th coefficient of an $N$-row matrix $\A$ is the $i+(j-1)N$-th element of $\vstack(\A)$. We use calligraphic letters to denote sets, and $|\mathcal{A}|$ is the cardinality of set $\mathcal{A}$. We use $\{\A[i]\}_{i=n}^{m}$ to represent a set that contains the collection of matrices $\A[i]$ for $i$ from $n$ to $m$. 
We use $f(n) = \Theta(g(n))$ to denote that there exist constants $c_1,c_2,n_0>0$ such that for all $n>n_0$, $c_1<\frac{f(n)}{g(n)}<c_2$
\cite{Knuth1976}.

\section{System Model}
\label{sec:model}

We assume a frequency selective, block fading, i.i.d. Rayleigh SIMO wideband channel with a single-antenna transmitter and $N$ antennas at the receiver. The transmitted signal bandwidth is divided in $B$ independent subcarriers separated by a bandwidth $\Delta f$ each, that experience independent frequency-flat Rayleigh distributed channel coefficients. We assume a block-fading channel in which, for each subcarrier, the channel coefficient remains the same for $L$ consecutive symbols, and takes a new i.i.d. value in each \textit{coherence slot}. We remark that the symbol period or channel use time in seconds is $\frac{1}{\Delta f}$, the channel coherence time in seconds is $\frac{L}{\Delta f}$, and $\Delta f$ may be interpreted as the ``coherence bandwidth'' as well as the subcarrier separation. Our channel model can be represented as a resource grid depicted in Fig. \ref{fig:channelOFDMAb}. Among others, this model is representative of contemporary multicarrier technologies, in which the OFDM symbol duration is shorter than the channel coherence time and pilots are transmitted once in each frame of several consecutive OFDM symbols. 

\begin{figure}[t]
  \centering
  \includegraphics[width=0.75\columnwidth]{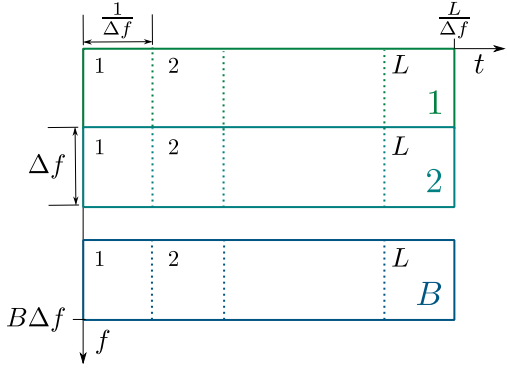}
  \caption{Block-fading frequency-selective channel output with bandwidth $B\Delta f$, coherence length $L$ and codeword time $L/\Delta f$.}
  \label{fig:channelOFDMAb}
\end{figure}

In each receive antenna $n$ the received signal coefficient on subcarrier $b$ at instant $\ell$ is given by
\begin{equation}
\label{eq:NBTchan}
  \begin{split}
    y_{n,\ell}[b]&=h_n[b]x_\ell[b]+z_{n,\ell}[b],\\
    &b\in\{0\dots B-1\},\\
    &n\in\{0\dots N-1\},\\
    &\ell\in\{0\dots L-1\}.
  \end{split}
\end{equation}
where the fading is normalized and independent for each $n$ and $b$, $h_n[b]\sim\mathcal{CN}(0,1)$. The Additive White Gaussian Noise (AWGN) is also normalized as $z_{n,\ell}[b]\sim\mathcal{CN}(0,1)$, and the input average power constraint is 
$$\Ex{}{\frac{1}{L}\sum_{b=1}^{B}\sum_{\ell=0}^{L-1}|x_\ell[b]|^2}\leq P.$$
Here, we use the single letter $P$ to represent a normalized power limitation. In practice $P$ is computed based on the actual transmitter power $P_T$, the channel mean gain $G$, the noise power spectral density $N_o$ and subcarrier bandwidth $\Delta f$ as $P=\frac{P_TG}{\Delta fN_o}$. Thus, $\frac{P}{B}=\frac{P_TG}{B\Delta fN_o}$ is the average SNR in each receive antenna.

Since channel inputs and outputs are independent for each frequency bin $b$, we can write the channel as a collection of $B$ frequency-flat SIMO subchannels 
\begin{equation}\label{eq:channel}
  \Y[b]=\h[b]\x^H[b]+\Z[b],\;\forall b\in\{0\dots B-1\}.
\end{equation}
where $\Y[b],\Z[b] \in\mathbb{C}^{N,L}$, $\h[b]\in\mathbb{C}^{N}$, and $\x[b]\in\mathbb{C}^{L}$ with some probability density function (p.d.f.) $p(\{\x[b]\}_{b=0}^{B-1})$. 

Even though we study the limit $L\to\infty$, we assume the ergodic mutual information is achieved by using the channel for a sufficiently long period of time of length $\gg L/B$, and we only apply constraints on the average power of the distributions. This permits to encode information over a large number of realizations of the block fading process and achieving the ergodic capacities defined as follows:

\begin{definition}
 The \textbf{coherent ergodic capacity} of the block fading Rayleigh channel as a function of $N$, $B$ and $L$ is
 \begin{equation}
 \label{eq:defCc}
 \begin{split}
    &C_c(N,B,L)=
    \\&    \sup_{p(\{\x[b]\}_{b=0}^{B-1})}\frac{\Ex{\{\h[b]\}_{b=0}^{B-1}}{\CInf{\{\x[b]\}_{b=0}^{B-1}}{\{\Y[b]\}_{b=0}^{B-1}}{\{\h[b]\}_{b=0}^{B-1}}}}{L/\Delta f}
    \end{split}
 \end{equation}
\end{definition}
\begin{definition}
 The \textbf{non-coherent ergodic capacity} of the block fading Rayleigh channel as a function of $N$, $B$ and $L$ is
 \begin{equation}
 \label{eq:defCn}    
 \begin{split}
    &C_n(N,B,L)=
    \\&    \sup_{p(\{\x[b]\}_{b=0}^{B-1})}\frac{\Ex{\{\h[b]\}_{b=0}^{B-1}}{\Inf{\{\x[b]\}_{b=0}^{B-1}}{\{\Y[b]\}_{b=0}^{B-1}}}}{L/\Delta f}
    \end{split}
 \end{equation}
\end{definition}
Here, the 
division by $L/\Delta f$
provides the unit conversion to express capacity in units of bits per second rather than bits per block. In this paper, we will study the asymptotic scaling when $B$, $N$ and $L$ all grow to infinity. Therefore, the division by the block length will play an important role in the scaling result. Generally, the capacity definition must consider joint encoding in all $B$ subcarriers, so the supremum is over the joint distribution $p(\{\x[b]\}_{b=0}^{B-1})$. However, in the case with no feedback and no CSIT we will show that, aside from power allocation, independent encoding in each subcarrier is sufficient.

In this paper we characterize the scaling of the functions $C_c(N,B,L)$ and $C_n(N,B,L)$ as the parameters $N$, $B$, and $L$ grow.
Since this is a multi-dimensional asymptotic, the point $(N,B,L)\to(\infty,\infty,\infty)$ can be approached from infinite ``directions''. We specify the relative growth speed of $N$, $B$, and $L$ as they jointly increase with scaling exponents.
\begin{definition}
 We define the \textbf{bandwidth exponent} as
\begin{equation}
  \epsilon = \lim_{B,\,N\rightarrow \infty}
  \frac{\log(B)}{\log(N)},
\label{eq:joint_scalingB}
\end{equation}
for a non-negative constant $\epsilon.$ This may equivalently be expressed as $B=\Theta(N^\epsilon).$
\end{definition}
\begin{definition}
The \textbf{block length exponent} is defined as 
\begin{equation}
  \tau = \lim_{L,\,N\rightarrow \infty}
  \frac{\log(L)}{\log(N)},
\label{eq:joint_scalingL}
\end{equation}
for a non-negative constant $\tau,$ or equivalently $L=\Theta(N^\tau).$
\end{definition}

Using these definitions, we can compare the scaling of $C_c(N,B,L)$ and $C_n(N,B,L)$ as a function of $N$, $\epsilon$ and $\tau$.

\section{Known Results}
\label{sec:knownres}

The comparison of the coherent and non-coherent channel capacities in \cite{8826443} showed that with a fixed $L$, i.e. $\tau=0$, EM achieves the capacity scaling of the non-coherent channel. However, for $\epsilon>\frac{1}{2}$ overspreading occurs, and the non-coherent channel capacity can only scale at most as $\sqrt{N}$, whereas the coherent channel capacity can reach linear scaling with $N$ if $\epsilon\geq 1$.

\begin{lemma}
\label{lem:gold1}
 The capacity scaling of a Coherent Block Fading Rayleigh i.i.d. wideband SIMO channel with $N$ receive antennas, bandwidth $B=\Theta(N^\epsilon)$ and $L=\Theta(N^\tau)$ is
 $$C_c(B,N,L)=\Theta(N^{\min(\epsilon,1)})$$
\end{lemma}
\begin{proof}
 Proof of\cite[Lemma 1]{8826443} holds also for $\tau>0$. 
\end{proof}

\begin{lemma}
\label{lem:gold2}
 The capacity scaling of a Non-Coherent Block Fading Rayleigh i.i.d. wideband SIMO channel with $N$ receive antennas, bandwidth $B=\Theta(N^\epsilon)$  and constant $L$ ($\tau=0$) is
 \begin{equation}
 \label{eq:CnL0}
  C_n(B,N,L)=\Theta(N^{\min(\epsilon,\frac{1}{2})})
 \end{equation}
\end{lemma}
\begin{proof}
 The full proof is given in \cite{8826443}. 
\end{proof}

The interest in EM techniques can be associated with the seminal work by Marzetta and Hochwald \cite{Marzetta1999}, which shows that the capacity-achieving input in the non-coherent channel is a product of an EM of unspecified density function, times an Isotropically Distributed Unitary Vector (IDUV) distribution.
\begin{lemma}
\label{lem:marzetta}
 \textbf{\cite[Theorem 2]{Marzetta1999}} The optimal distribution for a frequency-flat Rayleigh block fading $1\times N$ SIMO (sub)channel is $\x[b]=\sqrt{a[b]} \uu[b]$ where $a[b]$ and $\uu[b]$ are independent, $a[b]$ is a non-negative real number distribution and $\uu[b]$ is IDUV.
\end{lemma}

\begin{remark}
Comparing Lemmas \ref{lem:gold2} and \ref{lem:marzetta}, it would appear that in non-coherent channels of large dimensions, information is mostly carried by the EM part of the input, $a[b]$. In our work we will show that this interpretation may be an artifact caused by the assumption of a constant block length $L$ in the previous works. Indeed, this is the length of the IDUV $\uu[b]$, and when $L$ is relatively large compared to $B$ and $N$ the role of the ``input shape'', $\uu[b]$, in capacity scaling is not negligible.
\end{remark}

The achievable scheme for Lemma \ref{lem:gold2} is the EM technique proposed in \cite{Chowdhury2014a,8826443}. EM assumes that $M\leq B$ subchannels are used, each with an average transmitted power $\Ex{}{|\x[b]|^2}=LP/M$. The transmitter uses a real positive energy constellation to modulate the energy of the input, selecting a symbol from the set
$$a[b]\in \mathcal{C_M}=\left \{ 0, 2d, 4d,\cdots,\frac{2}{M} \right \},$$
and creating a transmitted signal by repeating the symbol $L$ times
$$\x[b]=\sqrt{a[b]}\one_{L}.$$

To perform detection, the receiver computes the statistic 
$$v[b]=\sum_{n=0}^{N-1}\left|\frac{1}{L}\sum_{\ell=0}^{L-1}y_{n,\ell}[b]\right|^2\sim (a[b]^2+ \frac{1}{L})\chi^2(2N)$$
and decides the symbol $\hat{a}[b]$ that is nearest to $v[b]- \frac{1}{L}$. As the transmitted symbol is repeated $L$ times, it is trivial to adjust the analysis of \cite[Theorem 2]{8826443} to the case with non-constant $L=\Theta(N^{\tau})$ as follows:
\begin{lemma}
 \label{lem:em}
 In EM with $M=\Theta(N^{\min(\epsilon,\frac{1}{2}+\tau)})$ and $d=\Theta(N^{-t})$ satisfying $\epsilon<t<\frac{1}{2}+\tau$, the probability of error vanishes as $N\to\infty$ and rate scales as $\Theta(N^{\min(\epsilon-\tau,\frac{1}{2})})$.
\end{lemma}
\begin{proof}
Since one data symbol is repeated $L$ times, the error probability upper bound computed in \cite[theorem 2]{8826443} vanishes for any inter-symbol distance that satisfies $t<\frac{1}{2}+\tau$ (this was $t<\frac{1}{2}$ with $\tau=0$ in the original proof). The constellation size is $\frac{N^{t}}{M}$, there are $M$ subchannels and a $1/L$ repetition code, so the rate is $\frac{M}{L}\log_2(\frac{N^t}{M})$. For any $\epsilon>0$ a rate scaling of $\Theta(N^{\min(\epsilon-\tau,\frac{1}{2})})$ is achievable, by choosing $M=\Theta(N^{\min(\epsilon,\tau+\frac{1}{2})})$ subchannels.
\end{proof}
\begin{remark}
 For $\tau=0$, the EM scheme achieves the non-coherent channel capacity scaling (Lemma \ref{lem:gold2}), and if we also assume $\epsilon<\frac{1}{2}$, the coherent channel capacity scaling is achieved (Lemma \ref{lem:gold1}).
\end{remark}

The main motivator to introduce an extension of Lemma \ref{lem:gold2} featuring a variable coherence block length is the exact calculation of the critical bandwidth for Rayleigh channels with fixed values of $B,N,L$, introduced in \cite{journals/twc/LozanoP12,fgomezUnified}, reproduced by the following lemma

\begin{lemma}
\label{lem:med1}
 In a Non-Coherent Block Fading Rayleigh i.i.d. wideband SIMO channel with constant $N$, $B$, and $L$, there is a ``critical bandwidth'' number $B_{crit}$, contained in the interval 
 \begin{equation}
 \label{eq:critB}
    2P\sqrt{\frac{\log\pi}{1+N}\frac{L}{\log L}}\leq B_{crit} \leq 2P\sqrt{(1+N)\log\pi\frac{L}{\log L}}
 \end{equation}
 and such that the achievable rate with non-peaky signaling is maximum when $B=B_{crit}$ and decreases for $B>B_{crit}$ (overspreading).
\end{lemma}
\begin{proof}
 Lemma 4 of \cite{fgomezUnified} reproduced with our notation.
\end{proof}

By observing the right hand side of \eqref{eq:critB} in Lemma \ref{lem:med1}, in comparison to the two capacity scaling regimes of \eqref{eq:CnL0} in Lemma \ref{lem:gold2}, it is clear that both results indicate that the maximum bandwidth scaling is proportional to $\Theta(\sqrt{N})$ when $L$ is fixed. By observing the role of $L$ in \eqref{eq:critB}, it is therefore natural to expect that, for non-constant $L$ that increases as $L=\Theta(N^\tau)$, an extension of Lemma \ref{lem:gold2} in which the critical bandwidth is proportional to $\Theta(\sqrt{NL})$ should exist. Although this is intuitive, the proof will be the main result of our paper. Significant mathematical differences versus \cite{journals/twc/LozanoP12,fgomezUnified,8826443} are required, described in Appendix \ref{sec:sub}.

\section{Main Result}
\label{sec:mainres}
\subsection{Capacity Scaling}
The main result in our paper introduces the effect of $L$ into the non-coherent capacity, fully characterizing its scaling.

\begin{theorem}
\label{th:cap}
 The capacity of a Non-Coherent Block Fading Rayleigh i.i.d. wideband SIMO channel with $N$ receive antennas, bandwidth $B=\Theta(N^\epsilon)$  and block length $L=\Theta(N^\tau)$ scales as
 $$C_n(B,N,L)=\Theta(N^{\min(\epsilon,\frac{1+\tau}{2},1)})$$
\end{theorem}
\begin{proof}
 This rate scaling is achievable using the PA scheme described in Sec. \ref{sec:pa}. Moreover, the same exponents are observed in the upper bound in Sec. \ref{sec:ub}. Therefore, the capacity scaling exponent is fully characterized.
\end{proof}

\begin{remark}
 When $\tau\geq 1$, the PA scheme can fully achieve the same capacity scaling as the coherent channel for any value of $\epsilon$. For $\tau<1$, the capacity scaling of the coherent channel is only matched by the non-coherent channel if its bandwidth is not too large, i.e. $\epsilon<\frac{1+\tau}{2}$.
\end{remark}

\begin{remark}
The comparison of Theorem \ref{th:cap} with Lemma \ref{lem:em} indicates that EM does not achieve the non-coherent capacity scaling for any value of $\epsilon$ if $\tau>0$. Even if both the EM and PA schemes can admit bandwidth exponents up to $\epsilon\leq\frac{1}{2}+\tau$ and  $\epsilon\leq\frac{1+\tau}{2}$, respectively, the EM scheme only achieves the rate exponent $\min(\epsilon-\tau,\frac{1}{2})$. Therefore, when $\tau>0$ the EM scheme displays poor spectral efficiency and does not match the rate scaling of the PA scheme, nor the non-coherent channel capacity.
\end{remark}

The interpretation of the main result is more intuitive by imagining the growth of rates in a channel with a very large but finite number of receive antennas ($N$), with a fixed coherence length ($L$), and with a very large bandwidth ($B\to\infty$) so that rate is always in the wideband ``power limited'' regime with overspreading. In essence, the main result warns us that the term ``scaling with $N$'' must be always considered \textit{in comparison to $L$}. We illustrate in Fig. \ref{fig:scalings} the rate versus the number of receive antennas for the special case of our result with $\epsilon\geq1$. In this special case the coherent capacity scales as $\Theta(N)$, the non-coherent capacity upper bound and the achievable rate with PA both scale as $\Theta(\min(\sqrt{NL},N))$, and the EM rate scales as $\Theta(\sqrt{N})$. We can thus see three regions for the non-coherent capacity scaling: First, the rates grow nearly linearly with $N$ when $N\ll L$, nearly matching the coherent channel capacity scaling. Second, when $L\ll N\ll L ^2$, the non-coherent channel capacity cannot continue to match the coherent channel, but the PA rate and the upper bound grow more than $\sqrt{N}$. In this regime PA achieves the non-coherent capacity scaling and outperforms EM. Finally, when $ L^2\ll N$, both the upper bound and PA achievable rate grow as $\sqrt{N}$, and the EM scheme achieves capacity scaling too. The prior scaling analysis assuming $L$ was a constant in \cite{8826443} reflected only this third regime. Nevertheless, in many channels $L$ is very large and the first two regimes have practical importance. For example, reasonably practical values of the coherence time $\frac{L}{\Delta f}=1 m s$ and delay spread $\frac{1}{\Delta f}=1\mu s$, would lead to $L=10^3$, which means that the third regime where EM is optimal is not achieved until $N> 10^6$. On the other hand, in a practical scenario with very high speed mobility where $L=2$ the third regime would appear quite soon, at $N\geq 4$. For $L=1$, the first two regimes would not exist at all.

\begin{figure}
 \centering
  \includegraphics[width=.55\columnwidth]{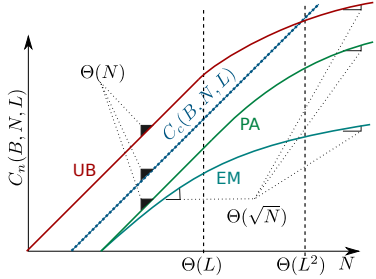}
  \caption{Scaling of rate vs $N$, with fixed $L$ and very large $B$}
  \label{fig:regionsS}
 \label{fig:scalings}
\end{figure}

\subsection{The Capacity with Only Energy Detection}

The energy-shape separation of the input in Lemma \ref{lem:marzetta}, and the ability of EM to achieve the non-coherent channel capacity scaling when $\tau=0$ in Lemma \ref{lem:gold2}, have arisen interest in encoding schemes that rely solely on encoding information using the instantaneous energy per transmitted block ($a[b]$ in lemma \ref{lem:marzetta}).
\begin{definition}
\label{def:oed}
We define an Only Energy Detection (OED) input as any input distribution that only encodes information in the input energy variable $a[b]\triangleq\|\x[b]\|^2$ according to the input decomposition $\x[b]=\sqrt{a[b]}\uu[b]$ proposed by Lemma \ref{lem:marzetta}. This definition comprehends both EM and also other modulations that satisfy this pattern, such as for example a Frequency Shift Keying (FSK) that holds the same frequency in the entire interval $\ell\in\{0\dots L-1\}$
\end{definition}
As an intermediate step of the converse proof in Theorem \ref{th:cap}, we have shown that no OED schemes can achieve a rate scaling greater than that of the EM scheme in Lemma \ref{lem:em}.

\begin{theorem}
\label{th:energycap}
In a Non-Coherent Block Fading Rayleigh i.i.d. wideband SIMO channel \eqref{eq:channel}, the OED Capacity defined as
\begin{equation*}
 \begin{split}
  &C_E(B,N,L)\triangleq 
\\&\sup_{p(\{a[b]\}_{b=0}^{B-1})}\frac{\Ex{\{\h[b]\}_{b=0}^{B-1}}{\Inf{\{a[b]\}_{b=0}^{B-1}}{\{\Y[b]\}_{b=0}^{B-1}}}}{L/\Delta f},\\
 \end{split}
\end{equation*}
for $N$ receive antennas, bandwidth $B=\Theta(n^\epsilon)$ and $L=\Theta(N^\tau)$, scales as
\begin{equation}
\label{eq:CEscaling}
 C_E(B,N,L)= \Theta(N^{\min(\epsilon-\tau,\frac{1}{2})})
\end{equation}
\end{theorem}
\begin{proof}
 The upper bound to the scaling of $C_E(B,N,L)$ is given in Appendix \ref{sec:eub}. The achievable rate is in Lemma \ref{lem:em}.
\end{proof}

\begin{remark}
  By the chain rule, it is clear that $C_E(B,N,L)\leq C_n(B,N,L)$. By comparing Theorems \ref{th:cap} and \ref{th:energycap}, they display equal scaling when $\tau=0$. Our PA achievable scheme has shown that the scaling factor $\sqrt{N}$ can be exceeded for $\tau>0$. However, Theorem \ref{th:energycap} is stronger than Lemma \ref{lem:gold2} in the sense that it proves that EM alone cannot exceed the square root limitation \textbf{even if we allow $\tau>0$}, and that no other OED schemes can either. The capacity scaling in Theorem \ref{th:cap} for $\tau>0$ cannot be achieved by simply adopting a different OED modulation other than EM in Lemma \ref{lem:em}. The rate scaling differences display a fundamental gap between OED Capacity and Non-Coherent Capacity scaling.
\end{remark}

A final question arises if non-coherent energy modulations can be improved at all. To this, we can answer favorably, but to achieve this we must not associate the concept of ``energy modulation'' with the variable $a[b]$ in Lemma \ref{lem:marzetta}. Indeed, in Section \ref{sec:fem} we propose a Fast Energy Modulation (FEM) scheme that transmits an energy constellation of $L$ independent i.i.d. symbols in the variable $\x[b]$. This allows to significantly improve the spectral efficiency, however a rate scaling limitation proportional to the square root of the number of antennas still arises.

\begin{lemma}
\label{th:fem}
The FEM scheme achieves a vanishing error probability as $N\to\infty$ with rate
\begin{equation}
\label{eq:RFEMscaling}
 R_{FEM}= \Theta(N^{\min(\epsilon,\frac{1}{2})})
\end{equation}
\end{lemma}
\begin{proof}
 The FEM scheme is analyzed in Section \ref{sec:fem}.
\end{proof}

\begin{remark}
 Theorem \ref{th:energycap} is valid for any OED scheme according to Definition \ref{def:oed}, following from the Marzetta-Hochwald enery-shape separation of the input in Lemma \ref{lem:marzetta}. The rate upper bound converse in Theorem \ref{th:energycap} does not apply to any non-coherent encoding scheme that does not meet Definition \ref{def:oed}. As in Lemma \ref{th:fem}, examples of non-OED schemes can be proposed for other non-coherent modulation techniques. For example, Theorem \ref{th:energycap} would not apply to an FSK scheme that allows to change the frequency in each instant $\ell$.
\end{remark}

We give a visual summary of all of our results in Fig. \ref{fig:Rexall}. For $\tau<1$, the capacity scaling of a coherent channel can be matched in a non-coherent channel only if the bandwidth is below a certain threshold, using either $a)$ a FEM scheme when $\epsilon<\frac{1}{2}$ or $b)$ a PA scheme when $\epsilon<\frac{1+\tau}{2}$. The latter is also the capacity scaling of the non-coherent channel. For $\tau\geq1$ the capacity scaling of the non-coherent channel equals that of the coherent channel for all bandwidths. The OED schemes designed with a strict interpretation of the energy-shape separation by Marzetta and Hochwald, such as the canonical EM scheme, cannot achieve the capacity scaling of the non-coherent channel at all $\forall \tau>0$.

\begin{figure}
\centering
\includegraphics{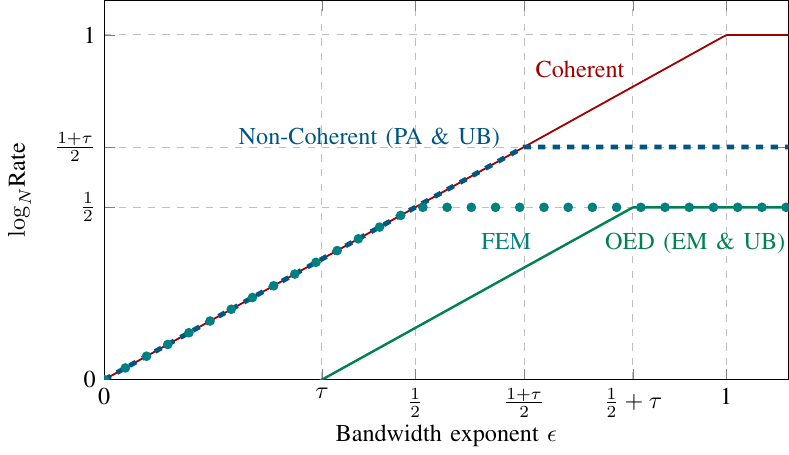}
\caption{Achievable rate exponents with $0<\tau<1$.}
\label{fig:Rexall}
\end{figure}

\section{Non-Coherent Capacity Upper Bound}
\label{sec:ub} 

\subsection{Subchannel Decomposition}

Each subchannel output depends only on the corresponding subchannel input, so we can separate the mutual information using the chain rule. Moreover the subchannels are i.i.d. and there is no CSI nor feedback, so time-diversity is exploited by the ergodic capacity average and frequency diversity is unnecessary. Therefore, we can fully decompose the rate as a sum of subchannel rate contributions:
\begin{equation}
\label{eq:decomp}
\begin{split}
 &C_n(B,N,L)=
 \\&
 \sup_{\substack{p(\{\x[b]\}_{b=0}^{B-1})\\s.t.\sum_{b=0}^{B-1}\Ex{}{|\x[b]|^2}\leq PL}}\frac{\Ex{\{\h[b]\}_{b=0}^{B-1}}{\Inf{\{\x[b]\}_{b=0}^{B-1}}{\{\Y[b]\}_{b=0}^{B-1}}}}{L/\Delta f}\\
 &\qquad=\frac{\Delta f}{L}\sup_{\substack{p(\{\x[b]\}_{b=0}^{B-1})s.t.\\\sum_{b=0}^{B-1}\Ex{}{|\x[b]|^2}\leq PL}} \sum_{b=0}^{B-1}\Ex{\h[b]}{\Inf{\x[b]}{\Y[b]}}\\
 &\qquad=\frac{\Delta f}{L}\sup_{\{\rho_b \}_{b=0}^{B-1}s.t.\sum_{b=0}^{B-1}\rho_b\leq PL} \sum_{b=0}^{B-1} f_{\mathrm{C}}(\rho_b)\\
\end{split} 
\end{equation}
where $\rho_b$ is the power allocated to subchannel $b$, and 
$$f_{\mathrm{C}}(\rho_b)\triangleq \sup_{p(\x[b])s.t.\\\Ex{}{|\x[b]|^2}\leq\rho_b}\Ex{\h[b]}{\Inf{\x[b]}{\Y[b]}}$$
is the subchannel capacity contribution in each of the identical subchannels as a function of $\rho_b$.

Since all subchannels are identical and $f_{\mathrm{C}}(\rho)$ is an increasing function of $\rho$, by the symmetry of the problem \eqref{eq:decomp} has a solution of the form
\begin{equation}
\label{eq:Moptim}
\begin{split}
 C_n(B,N,L)&=\frac{\Delta f}{L}\sup_{M\in\{1\dots B\}}Mf_C(PL/M).
\end{split} 
\end{equation}

In summary, to compute the capacity we must solve two problems: find the input distribution to achieve the maximum in $f_C(PL/M)$ for one subchannel, given a value of $M$; and then find the optimum number of subchannels that should be actively used, with equal power allocations to each. We reiterate that in this paper we assume a non-coherent channel without feedback in which the optimality of this equal power allocation is reasonable. However, in the cases with feedback or CSIT that we leave for future work, equal power allocation would likely be suboptimal.

%

\subsection{Energy-Shape Decomposition}

Since subchannels are i.i.d., hereafter we omit the index $[b]$ to simplify notation. Using Lemma \ref{lem:marzetta}, we need only to find the best energy p.d.f. $p(a)$, whereas the shape distribution is known to be uniform in the hyper sphere $p(\uu)=\frac{(\pi)^{L}}{(L-1)!}\delta(\|\uu\|^2-1)$. Using the chain rule we can write
\begin{equation}
\label{eq:esdivision}
\begin{split}
  f_C(\rho)
          &=\sup_{\substack{p(a)s.t.\\\Ex{}{|a|^2}\leq \rho}}\Ex{\h}{\Inf{a}{\Y}}+\Ex{\h,a}{\CInf{\uu}{\Y}{a}}\Big|_{\uu\sim \textrm{IDUV}}.
\end{split} 
\end{equation}

We recall that we must further maximize over the number of active subchannels $M$ in \eqref{eq:Moptim}. By adding the chain rule and Lemma \ref{lem:marzetta}, this results in
\begin{equation}
\label{eq:Mfunctions}
\begin{split}
 &C_n(B,N,L)=
 \\ &
 \sup_{\substack{1\leq M\leq B\\p(a)s.t.\\a>0,\Ex{}{a}\leq \frac{PL}{M}}}\frac{M\Delta f}{L}\left(\Ex{\h}{\Inf{a}{\Y}}+\Ex{\h,a}{\CInf{\uu}{\Y}{a}}\Big|_{\uu\sim \textrm{IDUV}}\right)\\
\end{split} 
\end{equation}
Moreover, we may upper bound the capacity \eqref{eq:Moptim} by performing separate optimizations in the first and second terms. Thus, we define the ``Only Energy Detection Capacity'' and the ``Shape Encoding Capacity'' such that their sum forms an upper bound to the non-coherent capacity as follows:

\begin{equation}
\label{eq:twotermCUB}
\begin{split}
 C_n(B,N,L)&\leq C_E(B,N,L)+C_S(B,N,L)\\
 \textnormal{with:}\\
 C_E(B,N,L)&=\sup_{\substack{1\leq M\leq B\\p(a)s.t.\\a>0,\Ex{}{a}\leq \frac{PL}{M}}}\frac{M\Delta f}{L}\Ex{\h}{\Inf{a}{\Y}}\\
 C_S(B,N,L)&=\sup_{\substack{1\leq M\leq B\\p(a)s.t.\\a>0,\Ex{}{a}\leq \frac{PL}{M}}}\frac{M\Delta f}{L}\Ex{\h,a}{\CInf{\uu}{\Y}{a}}\Big|_{\uu\sim \textrm{IDUV}}\\
\end{split} 
\end{equation}

\subsection{Scaling of the Sum of Upper Bounds}

We have upper bounded $C_E(B,N,L)\leq\Theta(N^{\min(\epsilon-\tau,\frac{1}{2})})$ in Theorem \ref{th:energycap}, which is proven in Appendix \ref{sec:eub}. For the second term we produce the following upper bound:
\begin{lemma}
 In a Non-Coherent Block Fading Rayleigh i.i.d. wideband SIMO channel \eqref{eq:channel} with $N$ receive antennas, bandwidth $B=\Theta(N^\epsilon)$ and $L=\Theta(N^\tau)$, the Shape-Encoding Capacity defined in \eqref{eq:twotermCUB}, scales at most as
\begin{equation}
\label{eq:CSscaling}
 C_S(B,N,L)\leq \Theta(N^{\min(\epsilon,\frac{1+\tau}{2},1)})
\end{equation}
\end{lemma}
\begin{proof}
 The proof is given in Appendices \ref{sec:sub}-\ref{app:fy}.
\end{proof}

Since for all $\epsilon>0$ we have that 
$$\min(\epsilon-\tau,\frac{1}{2})<\min(\epsilon,\frac{1+\tau}{2},1),$$
we observe that $C_E(B,N,L)\ll C_S(B,N,L)$ for sufficiently large $N$, and we can conclude that the non-coherent channel capacity scaling is upper bounded by 
$$C_n(B,N,L)\leq 2C_S(B,N,L) = \Theta(N^{\min(\epsilon,\frac{1+\tau}{2},1)}).$$
This concludes the proof of the converse of Theorem \ref{th:cap}.

\section{Pilot-Assisted Coherent Scheme}
\label{sec:pa}

We propose a simple pilot-assisted communication scheme on the channel \eqref{eq:NBTchan}. The design of the optimal number of pilot symbols and their allocated power in narrowband Rayleigh block fading MIMO channels was studied in \cite{Hassibi2003}. We assume $M$ subchannels are used simultaneously, with equal power allocation $\rho=\frac{PL}{M}$ each. Since we consider the SIMO case, one symbol in each subchannel $b$ must be used as a pilot \cite{Hassibi2003}. We assume the first symbol in each coherence block of each subchannel is the pilot $x_p[b]=x_0[b]$, and the remaining $L-1$ symbols are used to convey data $\{x_\ell[b]\}_{\ell=1}^{L-1}$. Moreover, we assume the power budget per subchannel is divided between pilots and data as $\Ex{}{|x_p|^2}=\alpha \rho$ and $\Ex{}{\frac{1}{L-1}\sum_{\ell=1}^{L-1}|x_\ell[b]|^2}=\frac{(1-\alpha)\rho}{L-1}$, for some fraction $\alpha\in[0,1]$.

We assume an independent estimation of each i.i.d. $h_{n}[b]$ using the first channel output in each antenna as
$$y_{p,n}[b]=h_n[b]x_p[b]+z_n[b]$$
where $z_n[b]\sim\mathcal{CN}(0,1)$. To maximize the achievable rate, the Minimum Mean Squared Error (MMSE) estimation $\hat{h}_{n}[b]=\frac{x^*_p}{|x_p|^2 +1}y_{p,n}[b]$ is employed \cite{Hassibi2003}. Using a constant amplitude pilot $|x_p[b]|=\sqrt{\alpha \rho}$ minimizes the variance of the MMSE estimator error, defined as $\tilde{h}_n[b]=h_{n}[b]-\hat{h}_{n}[b]\sim \mathcal{CN}(0,\varepsilon)$, where $\varepsilon=\frac{1}{1+\alpha\rho}$ \cite{Hassibi2003}.

In the data part, for $1\leq\ell<L-1$, we write
$$y_{n,\ell}[b]=\hat{h}_{n}[b]x_{\ell}[b]+\tilde{h}_{n}[b]x_{\ell}[b]+z_{n,\ell}[b],$$
where we assume i.i.d. symbols, so $\Ex{}{|x_{\ell}[b]|^2}=\frac{(1-\alpha)}{L-1}\rho.$

In addition, we assume the receiver employs Maximum Ratio Combining (MRC) treating $\hat{h}_{n}[b]$ as the true channel, producing
\begin{equation}
\begin{split}
r_{\ell}[b]&=\hat{\h}^H[b]\y_{\ell}[b]\\
&=\|\hat{\h}[b]\|^2x_{\ell}[b]+\left(\sum_{n=0}^{N}\hat{h}_{n}^*[b]\tilde{h}_{n}[b]\right)x_{\ell}[b]+z_{n,\ell}[b].
\end{split}
\end{equation}

Using a nearest neighbor decision is not optimal, but is sufficient to achieve at least the same rate as in an AWGN channel with an equivalent noise variance equal to the power of the second plus third terms of $r_{\ell}[b]$ \cite{Hassibi2003}. Since $(\h_{n}[b]|\hat{\h}_{n}[b])$ is distributed as $\mathcal{CN}(\hat{\h}_{n}[b],\varepsilon\I)$, the effective noise variance is $\varepsilon \frac{(1-\alpha)}{L-1}\rho+1$, and the effective signal energy is $\frac{(1-\alpha)}{L-1}\rho\|\h[b]\|^2(1-\varepsilon)$. Thus the achievable rate as a function of the perceived SNR at the symbol decider is
\begin{equation}
\label{eq:defRp}
R_{PA}(\alpha) \triangleq \frac{L-1}{L}M\log\left(1+\frac{\frac{(1-\alpha)}{L-1}\rho\|\h[b]\|^2(1-\varepsilon)}{1+\varepsilon \frac{(1-\alpha)}{L-1}\rho}\right) 
\end{equation}

In order to study the scaling as $N\to\infty$ we introduce the approximation $\|\h[b]\|^2\approx N$ for large $N$. We also define $\kappa\triangleq\frac{\rho}{L-1}$, and replace $\varepsilon=1/(1+\alpha \rho)$. With some rearrangement, this produces
%

\begin{equation}
\begin{split}
R_{PA}(\alpha) &\simeq\frac{L-1}{L}M\log(1+\frac{ (1-\alpha)\kappa N\rho\alpha}{1 + \alpha \rho + (1-\alpha)\kappa})\\
          &=\frac{L-1}{L}M\log(1+\underset{S_o}{\underbrace{\frac{N\kappa\rho}{1+\kappa}}}\underset{\gamma(\alpha)}{\underbrace{\frac{\alpha-\alpha^2}{\alpha\frac{\rho-\kappa}{1+\kappa}+1}}})
\end{split}
\end{equation}
which we can maximize by selecting $\alpha$ to maximize the decider-perceived SNR inside the logarithm, $S_o\gamma(\alpha)$. We perform a standard maximization of the function of the form $\gamma(\alpha)=\frac{\alpha-\alpha^2}{\alpha C_1+1}$, where the constant is $C_1=\frac{\rho-\kappa}{1+\kappa}$. This produces the following solution
$$\alpha^*=\frac{\sqrt{1+C_1}-1}{C_1}.$$

For convenience, we can rewrite the same result in the following forms as well
$$\alpha^*=\frac{1}{\sqrt{1+C_1}+1}=\frac{1}{\sqrt{\frac{1+\rho}{1+\kappa}}+1}.$$
In the third form, we can observe that, since $\rho>0$ and $\rho>\kappa>0$, one unconstrained global solution always satisfies $0\leq \alpha^* \leq 1$. Thus there is no need to introduce a Lagrangian multiplier to enforce such a constraint, and the global solution may be adopted directly.

Next we define a second constant $C_2=\sqrt{\frac{1+\rho}{1+\kappa}}=\sqrt{C_1+1}$. This allows to write, $\alpha^*=\frac{1}{C_2+1}$, $1-\alpha^*=\frac{C_2}{C_2+1}$, and $\alpha^*C_1+1=C_2$. Substituting we get
 \begin{equation}
 \begin{split}
    S_o\gamma(\alpha^*)&=S_o\frac{(1-\alpha^*)\alpha^*}{\alpha^* C_1+1}\\
                       &=S_o\frac{1}{(C_2+1)^2}\\
                       &=\frac{N\kappa\rho}{1+\kappa}\frac{1}{(\sqrt{\frac{1+\rho}{1+\kappa}}+1)^2}\\
                       &=\frac{N\kappa\rho}{(\sqrt{1+\rho}+\sqrt{1+\kappa})^2}\\
 \end{split}
 \end{equation}
Let $M=\Theta(N^{\mu})$ with $0\leq \mu\leq \epsilon$. We study the scaling of the decider-perceived SNR, $S_o\gamma(\alpha^*)$, as $N\to\infty$ for different values of $\tau$ and $\mu$. We recall that $\rho=\Theta(N^{\tau-\mu})$ and $\kappa=\Theta(N^{-\mu})$.
Thus, we must distinguish two different cases
\begin{enumerate}
 \item For $\mu\leq\tau$, we have $\rho\gg1\gg\kappa$. Therefore, $S_o\gamma(\alpha^*)\simeq N\kappa= \Theta(N^{1-\mu})$. For all $\epsilon$ we can choose $\mu\leq\min(\epsilon,\tau)$ and, if $\mu\leq 1$, \eqref{eq:defRp} scales as $R_{PA}=\Theta(N^\mu)$. Conversely, if $\tau>\mu>1$, $R_{PA}=\Theta(N)$.
 \item For $\mu>\tau$ we have $1\gg\rho\gg\kappa$. This means that, $S_o\gamma(\alpha^*)\simeq N\kappa\rho= \Theta(N^{1+\tau-2\mu})$. For all $\epsilon$ we can choose $\mu\leq\min(\epsilon,\frac{1+\tau}{2})$ and \eqref{eq:defRp} scales as $R_{PA}=\Theta(N^\mu)$. On the other hand, if we choose $\mu>\frac{1+\tau}{2}$, $R_{PA}$ scales as $\Theta(N^{1+\tau-\mu})\leq\Theta(N^{\frac{1+\tau}{2}})$. Clearly, making $\mu$ increase above $\frac{1+\tau}{2}$ is not beneficial, and the bandwidth overspreading constraint is $M\leq\Theta(\sqrt{NL})=\Theta(N^{\frac{1+\tau}{2}})$.
\end{enumerate}

When $\tau\leq 1$, $\frac{1+\tau}{2}$ is greater than $\tau$, with equality at $\tau=1$. Therefore, we can summarize all the PA achievable rate scaling with the following expression

$$R_{PA}=\Theta(N^{\min(\epsilon,\frac{1+\tau}{2},1)}),$$
concluding the proof of the achievability of Theorem \ref{th:cap}.

\section{Fast Energy Modulation (FEM)}
\label{sec:fem}
 We introduce an improved FEM scheme in order to compensate for the poor spectral efficiency of EM, which stems from its use of a $L\times$ repetition code. Although we propose FEM as a minor modification of EM, these two modulations echo the abundant literature on energy detection non-coherent schemes (see \cite{9445644} and references therein). To improve the spectral efficiency, we simply remove the repetition coding component as follows: In each block of $L$ symbols, the transmitter selects $L$ \textit{independent} energy symbols 
 $$a_0[b]\dots a_{L-1}[b]\in\mathcal{C}_M^L,$$
 and the transmitted signal is created as 
 $$x_{\ell}[b]=\sqrt{a_\ell[b]}.$$
 
 As in EM, we introduce an energy statistic for decoding, which now must take the form
 $$v_{\ell}[b]=\sum_{n=0}^{N-1}\left|y_{n,\ell}[b]\right|^2\sim (a_{\ell}[b]^2+1)\chi^2(2N).$$
 Here, we note that there is no repetition anymore, an independent statistic $v_{\ell}[b]$ is calculated and an independent nearest neighbor $\hat{a}_\ell[b]$ is decided for each value of $\ell$ and $[b]$. Since the channel remains constant for the duration of the block of $L$ symbols, yet the decoder is making symbol by symbol decisions, this decoder is clearly suboptimal. Nevertheless, this scheme is sufficient to guarantee that the error probability vanishes as $N\to\infty$ for the desired rates.
 
 To show this we need only to introduce a minor modification to the error analysis of \cite[theorem 2]{8826443} by considering the union bound over $L$ independently decoded symbols. Since a $\times L$ repetition code is no longer present as in EM, the inter-symbol distance requirement for FEM is that $d=\Theta(N^t)$ must now satisfy $t<\frac{1}{2}$. The constellation size is thus $\log_2(\frac{N^{t}}{M})\leq\log_2(N^{\frac{1}{2}-\epsilon})$. As a result, the total rate of FEM is $\Theta(M\log_2(N))$ with $M\leq\Theta(N^{\min(\epsilon,\frac{1}{2})})$. 
 
 For $\tau=0$, EM and FEM schemes achieve the same scaling; hence the use of EM in \cite{Chowdhury2014a} was sufficient to achieve the upper bound and prove the capacity scaling result. On the other hand, for $\tau>0$,
 EM cannot properly achieve the coherent capacity scaling, whereas FEM can for $\epsilon\leq\frac{1}{2}$.

 We remark that the FEM scheme we propose does \textbf{not} coincide with the Marzetta-Hochwald energy and shape decomposition of the input in Lemma \ref{lem:marzetta}. To prove the achievability of Theorem \ref{th:energycap}, in which we ultimately concluded that both terms in the Marzetta-Hochwald input decomposition are relevant, we have considered the EM scheme. FEM always achieves better rate exponent than EM, and it can achieve the capacity scaling for $\epsilon\leq \frac{1}{2}$. Nonetheless FEM does, in fact, encode information in the vector $\uu[b]$ of the decomposition $\x[b]=a[b]\uu[b]$, and is not a canonical OED scheme in the sense of Lemma \ref{lem:marzetta}.

\section{Simulation Results}
\label{sec:sim}

We demonstrate our results by simulating the progression of the Bit Error Rate (BER) and rate in the channel \eqref{eq:channel} as $N$ increases for different practical schemes. We simulate three achievable schemes: EM, FEM and PA, for different values of $\epsilon$ and $\tau$. To implement the simulation, the number of subchannels and block length must be integers, so we adopt $B=\lceil N^\epsilon\rceil$ and $L=\lceil N^\tau\rceil$. For EM and FEM, we adopt a binary constellation with symbols $\mathcal{C}^{EM}=\{0,\frac{2}{M}\}$. Recall that EM repeats the same symbol $L$ times in each channel coherence block, whereas FEM transmits sequences of $L$ i.i.d. symbols. For PA, we adopt a scaled BPSK constellation $\mathcal{C}^{PA}=\{-1,+1\}$ where the first symbol is a pilot scaled by a factor $PL\alpha^*$ and the remaining data symbols are scaled by a factor $\frac{PL}{L-1}(1-\alpha^*)$ according to the results in Section \ref{sec:pa}. We set $P=2$ so that the initial SNR per receive antenna is $3$ dB when $B=1$.

We depict in Fig. \ref{fig:ber} the BER for all schemes. We recall that the EM requires that $\epsilon<\frac{1}{2}+\tau$, PA requires $\epsilon<\frac{1+\tau}{2}$ and FEM requires the strictest condition $\epsilon<\frac{1}{2}$. For the case $\epsilon=0.3$ and $\tau=0$ (blue), all three conditions are met and all schemes are able to achieve a decreasing error probability as $N$ increases. It must be noted that $10^4$ symbols were simulated and the ``saw tooth'' shape of the curve is due to the use of the ceiling function to compute $M$ and not to a lack of sufficient bits for Monte Carlo approximation. For the case $\epsilon=0.6$ and $\tau=0$ (red), all conditions are unmet and all schemes experience a large error probability that does not decay as $N$ increases. Finally, when $\epsilon=0.6$ and $\tau=0.3$ (green), FEM experiences a large error probability, as the condition $\epsilon<\frac{1}{2}$ is not satisfied. On the other hand, the conditions $\epsilon<\frac{1}{2}+\tau$ and $\epsilon<\frac{1+\tau}{2}$ are both satisfied, and both EM and PA schemes display a decaying error probability as $N$ increases.

Next, we look at the \textbf{nominal} bit rates of each scheme in Fig. \ref{fig:Rb}, defined simply as the number of transmitted bits per second. For EM, the nominal rate is $R_b^{EM}=B\log\|\mathcal{C}^{EM}\|/L$. For FEM, the nominal rate is $R_b^{FEM}=B\log\|\mathcal{C}^{EM}\|$. And for PA, the nominal rate is $R_b^{EM}=B\log\|\mathcal{C}^{PA}\|(L-1)/L$.  For the case $\epsilon=0.3$ and $\tau=0$ (blue), we can see that all schemes achieve rates that scale as $N^{0.3}$. Moreover, as we saw in the BER curves, these rates correspond to a ``reliable communication'' in the sense that the BER vanishes even as the rate grows. For the case $\epsilon=0.6$ and $\tau=0$ (red), Fig. \ref{fig:Rb} indicates that all schemes transmit with rates that scale as $N^{0.6}$. However, as we saw in the BER curves, these are capacity-exceeding transmission rates in the sense that the communication at these rates is not reliable. Finally, when $\epsilon=0.6$ and $\tau=0.3$ (green), FEM transmits with a rate that scales as $N^{0.6}$ that is not reliable, EM communicates reliably but its rate scales only as $N^{0.3}$, and only PA can simultaneously guarantee that  the communication is reliable and that the rate scales as $N^{0.6}$.

The comparison of nominal rates and error probability in Figs. \ref{fig:ber} and \ref{fig:Rb} can be unified informally by using an ad-hoc metric we call the \textbf{Binary Symmetric Channel (BSC) equivalent reliable rate}, defined as $R_{BSCeq}=R_b(1-H(BER))$. This is, we compute the capacity of a BSC in which the error probability is equal to the empirical BER, and the bit rate is equal to the nominal bit rate of each modulation. Although this definition is not very rigorous, this ad-hoc metric allows to visualize the fraction of the nominal rate that is reliable in the simulation. We show the BSC equivalent reliable rates in Fig. \ref{fig:Gput}. Again, for  $\epsilon=0.3$ and $\tau=0$ (blue), all schemes achieve reliable rates that scale as $N^{0.3}$. We can also see that in the case with $\epsilon=0.6$ and $\tau=0$ all schemes achieve reliable rates that scale as $N^{0.5}$. Finally, for the case  $\epsilon=0.6$ and $\tau=0.3$, EM achieves only a rate scaling of $N^{0.3}$, FEM achieves a reliable rate scaling of $N^{.5}$, and only the PA scheme can achieve a reliable rate scaling of $N^{.6}=N^{\epsilon}$. Despite the non-rigorous definition of $R_{BSCeq}$, the concordance of all these results with our capacity analyses is complete and quite remarkable.

\begin{figure}
 \centering
 \includegraphics[width=.8\columnwidth]{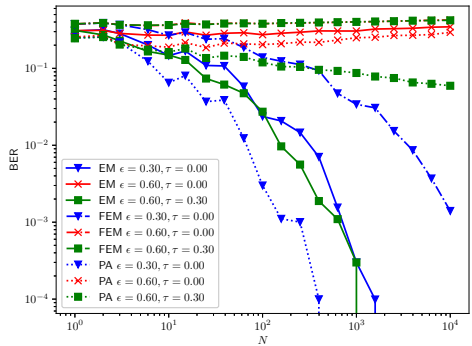}
 \caption{BER vs $N$ for different schemes.}
 \label{fig:ber}
\end{figure}

\begin{figure}
 \centering
 \includegraphics[width=.8\columnwidth]{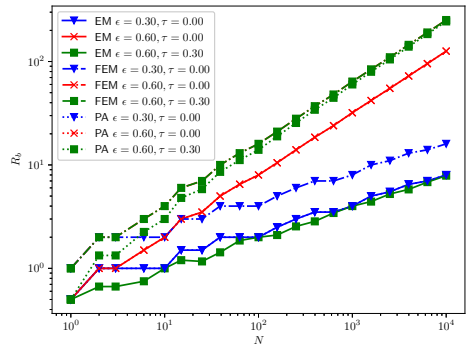}
 \caption{Bit rate vs $N$ for different schemes.}
 \label{fig:Rb}
\end{figure}

\begin{figure}
 \centering
 \includegraphics[width=.8\columnwidth]{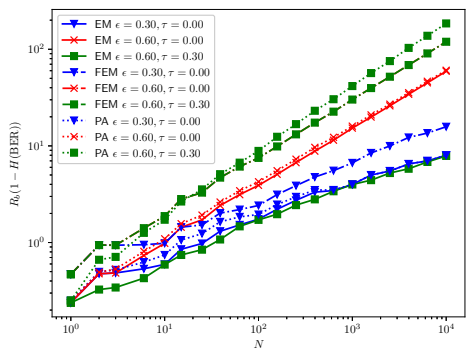}
 \caption{Equivalent BSC capacity vs $N$ for different schemes.}
 \label{fig:Gput}
\end{figure}

\section{conclusions}
\label{sec:conclusions}

In this paper we have studied the scaling of the capacity gap between coherent and non-coherent wideband SIMO Rayleigh block-fading channels, as a function of the bandwidth, number of receive antennas and channel coherence block length. The main insight of our paper is that the channel coherence block length plays a critical role in the capacity gap between coherent and non-coherent channels. Moreover, the maximum bandwidth allocation that can be used to avoid ``overspreading'' is also a key factor. In some prior formulations of the study of non-coherent channel capacity, Energy Detection techniques such as Energy Modulation have arisen as a promising encoding technique due to its capacity-like scaling with the number of antennas when the coherence block length is fixed. On the other hand, many practical technologies employ Pilot Assisted channel estimation and coherent receivers. Our result shows that Pilot Assisted schemes are always capacity-scaling-achieving in non-coherent channels, whereas Energy Detection schemes only achieve the capacity scaling in a more restricted sense that does not take into account the channel coherence block length.

Two important engineering messages emerge from our analysis. One is that as new technology standards continue to incorporate more spectrum, attention must be paid to the excess channel uncertainty caused by bandwidth overspreading, in which the capacity enters the wideband regime with power limited rates and the addition of even more degrees of freedom is not beneficial. The threshold to enter this regime scales as $\Theta\left(\min(\sqrt{NL},N)\right)$, which means that the limitation is more severe in high mobility scenarios. In addition, increasing the number of receive antennas can actually help push up the critical bandwidth threshold and permit to keep increasing wireless network rates by purchasing more spectrum.

The second key engineering lesson of our analysis is that the optimal choice between channel estimation or non-coherent encoding schemes seems to favor the first option in a significant number of cases. Only in very high mobility scenarios, which result in very short channel coherence length, the scaling analysis suggests that general Energy Detection schemes, and Energy Modulation in particular, are competitive in a rate-scaling sense. And even in those cases, both methods achieve equal scaling and a finer comparison of the rates beyond the scaling analysis should be performed before making the decision.

\appendices

\section{Discussion of the Channel Model}
\label{ap:chanmod}
The results in this paper are proven for an i.i.d. Rayleigh block-fading model. This classical fading model is extremely valued in wireless communications research for its excellent mathematical tractability while the results maintain some degree of proximity to the physical truth. Nevertheless, when considering a very large bandwdith or antenna array, engineers should be specially wary about the difference between this popular channel model and the physical reality of multi-path propagation. The usual justification of i.i.d. Rayleigh fading assumptions is rich scattering, that is, assuming a number of multipath reflections that is much higher than the number of taps, frequency bins, array elements, or both. This would motivate using the Central Limit Theorem (CLT) to characterize the coefficients, as a sum of many unknowns, as i.i.d. Gaussian. However, when the bandwidth is very large, the delays of multi-path reflections in the channel can be resolved individually. Likewise, the Angle of Arrival can be resolved in large antenna arrays with elements separated half a wavelength. These phenomena induce frequency correlation (delay sparsity) and spatial correlation (angular sparsity), respectively. In mmWave channels, advanced Saleh-Valenzuela sparse multipath models are more accurate \cite{Saleh1987}. In fact, many results show that in mmWave the CLT does not even justify the modeling of each reflection coefficient as Gaussian \cite{8844996,Mathew2016,TR38901,Coulson1998,bristolmeasurement}. The full-fledged physically-realist mmWave channel models in \cite{Mathew2016,TR38901} are, in fact, too cumbersome for analytical tractability. Theoretical analyses have resorted to different degrees of simplification ranging from sparse-Gaussian multipath \cite{journals/tit/TelatarT00,Mo2017} to almost-rich-scattering modified Rician or Rayleigh models that capture specific aspects of mmWave \cite{Du2017} such as oxygen absorption \cite{7896562}. More recently, there has been interest in ``holographic MIMO'' models that assume a very large number of antennas packed into a spatially constrained volume and spaced closer than $\lambda/2$ \cite{8437634,1386525}, in contrast to earlier mmWave array models with elements spaced $\lambda/2$ \cite{Mo2017,Mathew2016}, in which the antenna aperture goes to infinity with $N$.

The compromise between tractability and validity of the channel model is a common occurrence in capacity analises also outside mmWave. Lozano and Porrat indicate that their seminal critical bandwidth calculation holds for the ``wideband'' channel model, but may lose applicability as the bandwidth increases and the channel behaves as in the ``ultra-wideband'' model \cite{journals/twc/LozanoP12}. Ozgur, Leveque and Tse pointed out a similar weakness in network capacity scaling analysis \cite{Ozgur2007}. However, these authors have argued that the assumption of i.i.d. coefficients is justifiable when there is a range of the scaling analysis parameters in which simultaneously the number of elements is large and the channel assumptions hold. In our results, the argument is similar: our analysis is relevant for the range of values of the parameters $B$, $N$, and $L$ that are large-but-finite, and that are simultaneously sufficiently large for the big-O scaling results to dominate capacity, but not large enough that the channel model needs to be replaced. Abundant literature suggests that such a range of parameters exists. The seminal work by Medard and Gallager \cite{journals/tit/MedardG02} characterized the problem of overspreading without our assumptions on the channel model. Telatar and Tse showed that even a sparse multipath channel with a single path experiences overspreading when the reflection delay is not known a-priori \cite{journals/tit/TelatarT00}. Raghavan \textit{et al} found similar results for capacity scaling with bandwidth under sparse multipath channels \cite{Raghavan2007}. And the revision of the critical bandwidth results for mmWave by Ferrante \textit{et al} \cite{7896562} simplified the model to a multiplicative combination of rich Ricean fading, Bernouilli blockage and oxygen absorption.

In addition, the 5G \cite{3GPPNRoverall16} and WiFi6 \cite{khorov2019tutorial,8319416} standards employ OFDMA waveforms that treat each ``frame'' as an independent channel block realization with separate pilot estimation. Therefore, even when these standards operate on top of an underlying physical propagation that does not respond to the i.i.d. Rayleigh block fading model, the \textit{channel model assumption} that is implicit to the waveform design and engineering of many practical devices is indeed consistent with our model. In summary, while readers ought to be cautious about the channel model, our result insights are reasonably applicable to current wireless systems, and give analysts an important starting step toward more general results with advanced channel models in future work.

\section{OED Capacity Upper Bound}
\label{sec:eub}

Here we present an upper bound of the first term of \eqref{eq:Mfunctions}, $C_E(B,N,L)$, defined in \eqref{eq:twotermCUB}. This also proves the converse of Theorem \ref{th:energycap}.

Let us define the per-subchannel OED Capacity contribution function
$$f_E(\rho)\triangleq \sup_{p(a)s.t.a>0,\Ex{}{a}\leq \rho}\Ex{\h}{\Inf{a}{\Y}}$$
so that the OED Capacity satisfies
$$C_E(B,N,L)=\sup_{1\leq M\leq B}\frac{M\Delta f}{L}f_E(\frac{PL}{M}).$$

We can use $\x^H=\sqrt{a}\uu^H$ as defined in Lemma \ref{lem:marzetta} to define the auxiliary random variable
\begin{equation}
\label{eq:mkEqChan}
\vv=\sqrt{a}\h+\z_{a},
\end{equation}
where $\z_a=\Z\uu$ is i.i.d. AWGN of dimensions $N\times 1$. Since $\Y=\vv\uu^H+\Z(\I-\uu\uu^H)$, we verify that $p(\Y|\vv,a)=p(\Y|\vv)$ and therefore $a\to \vv \to\Y$ is a Markov chain satisfying 
$$f_E(\rho)\leq \sup_{p(a)s.t.a>0,\Ex{}{a}\leq \rho}\Inf{a}{\vv}\triangleq g_E(\rho).$$

Finally, we have that \eqref{eq:mkEqChan} is a SIMO Rayleigh fading channel with block length of $L=1$ and with power constraint $\Ex{}{a}\leq \frac{PL}{M}$, so we can apply \cite[Lemma 3]{8826443} to $g_E(\rho)$
\begin{lemma}
\label{lem:ubenergy}
(\cite[Lemma 3]{8826443}) For $\rho=\frac{PL}{M}=\Theta(N^{\frac{1}{2}+\alpha})$ with any $\alpha>0$, the upper bound 
$g_E(\rho)\leq \Theta(\frac{1}{N^{2\alpha}})$
is satisfied. 
\end{lemma}
\begin{proof}
 Same as \cite{8826443} with power constraint $\rho=\frac{PL}{M}$.
\end{proof}

As a result, when $B<\Theta(N^{\frac{1}{2}+\tau})$, $f_E(\frac{P}{M})$ is upper bounded by the coherent capacity of the channel \eqref{eq:mkEqChan}, and the optimal number of active subchannels is $M=B$. When $B\geq\Theta(N^{\frac{1}{2}+\tau})$, $g_E(\frac{P}{M})$ is upper bounded by Lemma \ref{lem:ubenergy} and the transmitter must not overspread the transmitted power in more than $M=\Theta(\sqrt{N}L)$ subchannels. Thus in general the optimal $M$ is $\Theta(N^{\min(\epsilon,\frac{1}{2}+\tau)})$ and, substituting into $\frac{M}{L}g_E(\frac{PL}{M})$, we get \eqref{eq:CEscaling}, completing the proof of Theorem \ref{th:energycap}.

%

\section{Shape Encoding Capacity Upper Bound}
\label{sec:sub}

Here we present an upper bound of the second term in \eqref{eq:Mfunctions}, $C_S(B,N,L)$, defined in \eqref{eq:twotermCUB}. We recall the definition of the conditional mutual information
$$\Ex{a}{\CInf{\uu}{\Y}{a}}=\Ex{a}{\Ex{\Y,\uu}{\log \frac{p(\Y|\uu,a)}{p(\Y|a)}}}.$$
Since we cannot compute this information exactly, we will proceed by first upper bounding the information function defined as
$\CInf{\uu}{\Y}{a}\leq \Phi(a)$
for given values of $a$. And finally we will maximize $\Ex{a}{\Phi(a)}$ to compute an upper bound of the scaling of $C_S(B,N,L)$.

Thanks to the fact that $\uu$ is unitary and $\uu,a$ are known, we compute the distribution $p(\Y|\uu,a)$ explicitly in Appendix \ref{app:fyu}, resulting in
$$p(\Y|\uu,a)=(1+ a)^{-N}(2\pi)^{-NL}\textnormal{exp}\left[-(\|\Y\|^2-\frac{a\|\Y\uu\|^2}{1+a})\right]$$

For the denominator, we settle for simply computing a function that lower bounds $p(\Y|a)$ as a function of $a$ (Appendix \ref{app:fy}). This function is \textbf{not} a p.d.f., and is expressed as
$$p(\Y|a)\geq(1+ a)^{-N}(2\pi)^{-NL}\textnormal{exp}\left[-(\|\Y\|^2-\frac{a}{1+a}\frac{\|\Y\|^2}{L})\right]$$

Combining the above, we have the following upper bound for the mutual information
\begin{equation}
\begin{split}
\CInf{\uu}{\Y}{a}
&\leq \frac{a}{1+a}\left(\Ex{\Y,\uu}{\|\Y\uu\|^2}-\frac{\Ex{\Y}{\|\Y\|^2}}{L}\right)\\
\end{split}
\end{equation}

To continue we need to calculate the average energies $\Ex{}{\|\Y\uu\|^2}$ and $\Ex{}{\|\Y\|^2}$. Noting that $\Y\uu=\sqrt{a}\h\uu^H\uu+\Z\uu=\sqrt{a}\h+\z'$,  where $\z'\sim\mathcal{CN}(0,\I_N)$, we get that $\Ex{}{\|\Y\uu\|^2}=N(1+a)$. With similar reasoning we can deduce that $\Ex{}{\|\Y\|^2}=N(L+a)$ from the fact that $\|\h\uu^H\|^2=\|\h\|^2$ and that $\Ex{}{\|\Z\|^2}=L\Ex{}{\|\z'\|^2}$. Thus we arrive at the following upper bound to the mutual information

\begin{equation}
 \label{eq:ub1}
 \CInf{\uu}{\Y}{a}\leq\Phi_1(a)\triangleq 
 \frac{a^2}{1+a}N\left(1- \frac{1}{L}\right).
\end{equation}

Since $\Phi_1(a)$ is convex, unless $a$ is subject to a peak constraint, introducing the upper bound $\Phi_1(a)$ into the average and maximizing with regard to $p(a)$ obtains an upper bound to the Shape Encoding Capacity that is too loose. In particular
$$\sup_{p(a)}\Ex{a}{\CInf{\uu}{\Y}{a}}<\sup_{p(a)}\Ex{a}{\Phi_1(a)}=\rho N (1- \frac{1}{L}).$$
where the supremum is achieved by the peaky distribution 
$$p(a)=\lim_{\delta\to0}\begin{cases}\rho/\delta&\textnormal{w.p. }\delta\\0&\textnormal{otherwise}\end{cases}.$$

In other words, while $\Phi_1(a)$ can upper bound $\CInf{\uu}{\Y}{a}$ for small values of $a$, we cannot use this upper bound to study the optimal distribution of $p(a)$ when there are no peak constraints. Therefore, as a correction factor to $\Phi_1(a)$, let us use the capacity of a AWGN channel of power $a$ as a second upper bound to the mutual information
\begin{equation}
 \label{eq:ub2}
 \begin{split}
 \CInf{\uu}{\Y}{a}&\stackrel{a}{\leq}\sup_{p(\x):\Ex{}{\|\x\|^2}\leq a}\Inf{\x}{\Y}\\
    &\stackrel{b}{\leq}L\log(1+aN)\\
    &\stackrel{c}{\triangleq} \Phi_2(a),
    \end{split}
\end{equation}
where $a)$ stems from the fact that $\uu$ is an IDUV, so get that $\Ex{}{\|a\uu\|^2}=a$, and therefore the inequality with the supremum holds. Next, $b)$ upper bounds the non-coherent mutual information by the coherent capacity with power constraint $a$. And $c)$ simply writes this second upper bound as a function depending on $a$.

Finally, we can properly conduct an input distribution optimization with regard to the  ``useful'' upper bound defined as the minimum of $\Phi_1(a)$ and $\Phi_2(a)$
\begin{equation}
 \label{eq:ub3}
 \begin{split}
 \CInf{\uu}{\Y}{a}&\leq \Phi(a)\\
 &
 \triangleq 
 \min\left[\Phi_1(a),\Phi_2(a)\right]\\
 &
 = \min\left[\frac{a^2}{1+a}N\left(1-\frac{1}{L}\right),L \log(1+aN)\right]
 \end{split}
\end{equation}
Leading to
$$C_S(B,N,L)\leq\sup_{M\in\{1\dots B\}}\frac{M\Delta f}{L}\sup_{p(a)}\Ex{a}{\Phi(a)}.$$

The shape of $\Phi(a)$ is very particular. Let us define the point where the two sides of the minimum are equal as $a_o$ satisfying
\begin{equation}
\label{eq:defa0}
 \frac{a_o^2}{1+a_o}N\left(1-\frac{1}{L}\right)=L \log(1+a_oN).
\end{equation}
For $a<a_o$, $\Phi(a)$ is convex, and for $a>a_o$, $\Phi(a)$ is concave. Therefore, we can solve the optimization separately in the concave and convex regions. To this end, we define the auxiliary variable $\Upsilon=\int_{0}^{a_o}p(a)da$, and the auxiliary distributions $q(a)=\frac{p(a|a\leq a_o)}{\Upsilon}$, and $g(a)=\frac{p(a|a> a_o)}{1-\Upsilon}$. These auxiliary variables enable the following separation of the original problem:
\begin{equation}\begin{split}
\sup_{p(a)}\Ex{a}{\Phi(a)} = \sup_{\Upsilon}\Big(&\Upsilon \sup_{q(a)}\Ex{a\leq a_o}{\Phi(a)}+\\
                                                      &+(1-\Upsilon) \sup_{g(a)}\Ex{a> a_o}{\Phi(a)}\Big).
\end{split}\end{equation}
Here, we separated the original problem in three subproblems, where one is purely convex and another is purely concave. Therefore we can consider Jensen's inequality to find the optimum distributions $q(a)$ and $g(a)$, and then complete the final result by finding the optimum value of $\Upsilon$ and constructing the final density function as $p^*(a)=\Upsilon^* q^*(a)+(1-\Upsilon^*)g^*(a)$. This leads to the following results:
\begin{enumerate}
 \item If $a_o\geq\rho$, then $\Upsilon^*=1$ and  
 \begin{equation}
  \label{eq:suppacase1}p^*(a)=q^*(a)=\begin{cases}
          a_o& \textnormal{w.p.} \frac{\rho}{a_o}\\
          0&\textnormal{otherwise}
         \end{cases}
  \end{equation}
   producing
 \begin{equation}
  \label{eq:supIacase1}
   \sup_{p(a)}\Ex{a}{\Phi(a)}=\rho \frac{a_o}{1+a_o} N(1-\frac{1}{L})
 \end{equation}
 \item If $a_o<\rho$, then $\Upsilon^*=0$, and 
 \begin{equation}
  \label{eq:suppacase2}
    p^*(a)=g^*(a)=\delta(a-\rho),
 \end{equation}
 producing
 \begin{equation}
  \label{eq:supIacase2}
  \sup_{p(a)}\Ex{a}{\Phi(a)}=L\log(1+\rho N)
 \end{equation}
\end{enumerate}

Since solving \eqref{eq:defa0} can be difficult, to complete our analysis we obtain a simple scaling characterization of $a_0$.

\begin{lemma}
 If $\tau\leq1$ the solution to \eqref{eq:defa0} scales as
 $$a_o=\Theta(\sqrt{\frac{L}{N}})$$
 And if $\tau>1$ the solution to \eqref{eq:defa0} scales as
 $$a_o=\Theta(\frac{L}{N})$$
\end{lemma}
\begin{proof}
 By contradiction, let us define $a_o'=\Theta(N^\alpha)$ and $\Phi_2(a_o')=\theta(N^{\tau}\log(N))$. 
 
 If $\alpha\leq0$, then $\Phi_1(a_o')=\theta(N^{2\alpha-1})$. Thus for $\alpha<\frac{\tau-1}{2}$ we have $\Phi_1(a_o')\ll\Phi_2(a_o')$ and if $\alpha>\frac{\tau-1}{2}$ then $\Phi_1(a_o')\gg\Phi_2(a_o')$. Therefore, $a_o'$ cannot be a solution to \eqref{eq:defa0} if $\alpha\neq\frac{\tau-1}{2}$, where $\alpha<0\Leftrightarrow \tau<1$. 
 
 Likewise, if $\alpha>0$, then $\Phi_1(a_o')=\theta(N^{\alpha-1})$, and $a_o'$ cannot be a solution to \eqref{eq:defa0} if $\alpha\neq \tau-1$, where $\alpha>0\Leftrightarrow \tau>1$.
\end{proof}

\begin{figure*}[!b]
\normalsize
\vspace*{4pt}
\hrulefill
\setcounter{MYtempeqncnt}{\value{equation}}
\setcounter{equation}{31}

\begin{equation}
\label{eq:fYua}
 \begin{split}
  p(\Y|\uu,a)&=|(\I_{L}+ a\uu^*\uu^T)\otimes\I_{N}|^{-NL}(2\pi)^{nL}\textnormal{exp}\left[-\vstack(\Y)^H\left((\I_{L}+ a\uu^*\uu^T)\otimes\I_{N}\right)^{-1}\vstack(\Y)\right]\\
  &\stackrel{(a)}{=}|\I_{N}|^{-N}|\I_{L}+ a\uu^*\uu^T|^{-L}(2\pi)^{-NL}\textnormal{exp}\left[-\vstack(\Y)^H\left(\left(\I_{L}+ a\uu^*\uu^T\right)^{-1}\otimes\left(\I_{N}\right)^{-1}\right)\vstack(\Y)\right]\\
  &\stackrel{(b)}{=}|\I_{L}+ a\uu^*\uu^T|^{-L}(2\pi)^{-NL}\textnormal{exp}\left[-\vstack(\Y)^H\left(\left(\I_{L}+ a\uu^*\uu^T\right)^{-1}\otimes\I_{N}\right)\vstack(\Y)\right]\\
  &\stackrel{(c)}{=}(1+ a)^{-N}(2\pi)^{-NL}\textnormal{exp}\left[-\vstack(\Y)^H\left(\left(\I_{L}-\frac{a}{1+a\uu^T\uu^*}\uu^*\uu^T\right)\otimes\I_{N}\right)\vstack(\Y)\right]\\
  &\stackrel{(d)}{=}(1+ a)^{-N}(2\pi)^{-NL}\textnormal{exp}\left[-(\|\Y\|^2-\frac{a}{1+a}\|\Y\uu\|^2)\right]\\
 \end{split}
\end{equation}

\setcounter{equation}{\value{MYtempeqncnt}}
\end{figure*}

%
%

Equipped with the scaling of $a_o$, recalling that $\rho= \frac{PL}{M}$, and using \eqref{eq:supIacase1} and  \eqref{eq:supIacase2}, we get the following cases:
\begin{itemize}
 \item For $\tau\leq 1$, 
 \begin{itemize}
    \item if $M\leq\Theta(\sqrt{NL})$ we get $a_o\ll\rho$, which leads to $\frac{M}{L}\Ex{a}{\Phi(a)}=\frac{M}{L}\Phi_2(\rho)=M\log(1+\rho N)$.
    \item if $M\geq\Theta(\sqrt{NL})$ we get $a_o\gg\rho$, leading to $\frac{M}{L}\Ex{a}{\Phi(a)}=\frac{M}{L}\frac{\rho}{a_o}\Phi_1(\rho)=\frac{M}{L}a_o\rho N(1-\frac{1}{L})=\Theta(\sqrt{LN})$.
  \end{itemize}
  Thus, for $\epsilon<\frac{1+\tau}{2}$, the optimum number of subchannels satisfies $M=B$. And for the case $\epsilon\geq\frac{1+\tau}{2}$, all choices of $M$ satisfying $M\geq\Theta(N^{\frac{1+\tau}{2}})$ are equally optimal in a scaling sense. Although different choices of $M\geq\Theta(N^{\frac{1+\tau}{2}})$ can overspread the bandwidth, this neither increases nor decreases the rate. Increasing $M$ above $\Theta(N^{\frac{1+\tau}{2}})$ merely increases the peakyness of the optimal distribution $p^*(a)$ \eqref{eq:suppacase1} in such a way that the time-frequency spreading of power remains constant and the rate scaling is unchanged.
  \item Finally, when $\tau>1$ the bandwidth overspreading threshold condition associated with $a_o\ll\rho$ scales as $M\leq\Theta(N)$, and the rate is $\Theta(N^{\min(\epsilon,1)})$.
\end{itemize}

Putting everything together produces the final Shape Encoding Capacity scaling upper bound 
$$C_S(B,N,L) \leq \Theta(N^{\min(\epsilon,\frac{1+\tau}{2},1)}).$$

\section{Distribution of $p(\Y|\uu,a)$}
\label{app:fyu}

In this section we compute the probability density $p(\Y|\uu,a)$. The Kro\"enecker product and matrix vector stacking function satisfy the relation $\vstack(\h\uu^H)=\uu^*\otimes\h$. When $\uu$ is a known vector, $(\vstack(\h\uu^H)|\uu)$ is a $NL$-dimension Gaussian-distributed vector with a singular covariance matrix $\Sigma_{\vstack(\h\uu^H)}= a\uu^*\uu^T\otimes\I_{N}$, and its p.d.f. is undefined. Fortunately, the addition of the nose term in $(\Y|a,\uu)=a\h\uu^H+\Z$ introduces independent noise terms and, therefore, $p(\Y|\uu,a)$ can be written as the p.d.f. of the $NL$-dimension Gaussian-distributed vector $(\vstack(\Y)|a,\uu) \sim \mathcal{CN}(0,\Sigma_{\vstack(\Y)|\uu,a})$, where the covariance matrix is full rank with value $\Sigma_{\vstack(\Y)|\uu,a}=\I_{NL}+ a\uu^*\uu^T\otimes\I_{N}=(\I_{L}+ a\uu^*\uu^T)\otimes\I_{N}$. 

For convenience, we write a more compact expression of the p.d.f.
\stepcounter{equation} in \eqref{eq:fYua},
where $(a)$ comes from the distributive properties of the determinant and inverse of Kro\"enecker products, $(b)$ substitutes the determinant and inverse of the identity matrix, $(c)$ applies the Sylvester rule for the determinant, and the Sherman-Morrison formula for the matrix inverse, and $(d)$ comes from $\uu^T\uu^*=1$, $\vstack(\Y)^H\vstack(\Y)=\|\Y\|^2$, and a careful reorganization of the matrix product $\vstack(\Y)^H(\uu^*\uu^T\otimes\I_{N})\vstack(\Y)=\tr\{(\Y\uu)^H\Y\uu\}=\|\Y\uu\|^2$.

\begin{figure*}[!b]
\normalsize
\vspace*{4pt}
\hrulefill
\setcounter{MYtempeqncnt}{\value{equation}}
\setcounter{equation}{32}

\begin{equation}
\label{eq:fYa}
\begin{split}
p(\Y|a)&=\int_{|\uu|^2=1}p(\Y|\uu,a)p(\uu)d \uu\\
&=\int_{|\uu|^2=1}(1+ a)^{-N}(2\pi)^{-NL}\textnormal{exp}\left[-(|\Y|^2-\frac{a}{1+a}|\Y\uu|^2)\right]\frac{\Gamma(L)}{\pi^{L}}d \uu\\
&=(1+ a)^{-N}(2\pi)^{-NL}\textnormal{exp}\left[-|\Y|^2\right]\int_{|\uu|^2=1}\frac{\Gamma(L)}{\pi^{L}}\textnormal{exp}\left[\frac{a}{1+a}|\Y\uu|^2\right]d \uu\\
&=(1+ a)^{-N}(2\pi)^{-NL}\textnormal{exp}\left[-|\Y|^2\right]\Ex{\uu}{\textnormal{exp}\left[\frac{a}{1+a}|\Sm\uu|^2\right]}\\
 \end{split}
\end{equation}

\setcounter{equation}{\value{MYtempeqncnt}}
\end{figure*}

\section{Lower bound of $p(\Y|a)$}
\label{app:fy}

The p.d.f. of the channel output may be computed from the conditional as $p(\Y|a)=\Ex{\uu}{p(\Y|a,\uu)}$. We note that for the particular case $L=1$ we have that $\uu$ is just a random phase, producing $\|\Y\uu\|^2=\|\Y\|^2$, this would lead to $p(\Y|a)|_{L=1}=p(\Y|a,\uu)|_{L=1}$ and to $\CInf{\uu}{\Y}{a}=0$. We therefore focus on the case $L>1$, when $p(\Y|a)$ is not Gaussian distributed and we are forced to come up with a lower-bounding strategy on  $p(\Y|a)$ in order to upper bound $\CInf{\uu}{\Y}{a}$.

The p.d.f. of $\uu$ is constant in the $L$-dimensional sphere, with value $\Gamma(L)/\pi^{L}$, and thus remains unchanged when $\uu$ is transformed by multiplication with any unitary matrix $p(\V\uu)=p(\uu)$. Particularly, we use the SVD of $\Y=\U^H\Sm\V$ where $\U,\V$ are unitary, $\Sm$ is the diagonal matrix containing the singular values of $\Y$, and $\uu^H\Y^H\Y\uu=\uu^H\V^H\Sm^2\V\uu$. Therefore
we get \stepcounter{equation}\eqref{eq:fYa}, shown on the bottom of this page.

Now, for the purpose of just upper bounding the mutual information, we apply Jensen's inequality in the exponential $\Ex{\uu}{\textnormal{exp}\left[\frac{a}{1+a}\|\Sm\uu\|^2\right]}\geq\textnormal{exp}\left[\frac{a}{1+a}\Ex{\uu}{\|\Sm\uu\|^2}\right]$. With this we obtain a lower bound of the function $p(\Y|a)$ that is not necessarily a p.d.f.

\begin{equation}
\begin{split}
p(\Y|a)&\geq(1+ a)^{-N}(2\pi)^{-NL}\textnormal{exp}\left[-(\|\Y\|^2-\frac{a\Ex{\uu}{\|\Sm\uu\|^2}}{1+a})\right]\\
 \end{split}
\end{equation}

Since $\Sm$ is diagonal, $\|\Sm\uu\|^2=\sum_{\ell=1}^{L}\left|s_\ell u_{\ell}\right|^2$, producing
$$\Ex{\uu}{\|\Sm\uu\|^2}
=\sum_{\ell=1}^{L}s_\ell^2\Ex{\uu}{|u_\ell|^2}
=\frac{1}{L}\sum_{\ell=1}^{L}s_\ell^2
=\frac{\|\Y\|^2}{L}.$$

Therefore 

\begin{equation}
\begin{split}
p(\Y|a)&\geq(1+ a)^{-N}(2\pi)^{-NL}\textnormal{exp}\left[-\|\Y\|^2(1-\frac{a}{1+a}\frac{1}{L})\right]\\
 \end{split}
\end{equation}

\end{document}